\PassOptionsToPackage{dvipsnames}{xcolor}

\documentclass[12pt,reqno]{article}

\usepackage{amsfonts,latexsym,amsthm,amssymb,amsmath,amscd,euscript,tikz, bm, mathtools, commath,centernot,float}
\usepackage{algorithm}
\usepackage{algpseudocode}
\usepackage{fullpage}
\usepackage{amsmath}
\usepackage{dsfont}
\usepackage{enumerate}
\usepackage{enumitem}
\usepackage{hyperref}
\usepackage{bbm}
\usepackage{comment}
\usepackage{xcolor}
\usepackage{tikz} 
\usetikzlibrary{arrows}
\usetikzlibrary{positioning,calc}

\usepackage{authblk}
\usepackage{natbib}

\hypersetup{colorlinks=true,citecolor=blue,urlcolor =black,linkbordercolor={1 0 0}}
\hypersetup{
  linkcolor=cyan 
}

\allowdisplaybreaks[1]

\newtheorem{theorem}{Theorem}
\newtheorem{proposition}{Proposition}
\newtheorem{lemma}{Lemma}

\newtheorem{assumption}{Assumption}

\theoremstyle{definition}
\newtheorem{defn}{Definition}
\newtheorem{example}{Example}

\newcommand{\cond}{\mid}
\newcommand{\indep}{\perp \!\!\! \perp}

\theoremstyle{remark}
\newtheorem*{remark}{Remark}

\title{A general condition for bias attenuation by a nondifferentially mismeasured confounder}

\author[,1]{Jeffrey Zhang\thanks{Corresponding author. Email: \texttt{jzhang17@wharton.upenn.edu}. }}
\author[1]{Junu Lee} 
\affil[1]{{\normalsize Department of Statistics and Data Science, University of Pennsylvania, USA}} 
\date{}

\begin{document}

\maketitle
\begin{abstract}

    In real-world studies, the collected confounders may suffer from measurement error. Although mismeasurement of confounders is typically unintentional---originating from sources such as human oversight or imprecise machinery---deliberate mismeasurement also occurs and is becoming increasingly more common.  For example, in the 2020 U.S. Census, noise was added to measurements to assuage privacy concerns. Sensitive variables such as income or age are oftentimes partially censored and are only known up to a range of values. In such settings, obtaining valid estimates of the causal effect of a binary treatment can be impossible, as mismeasurement of confounders constitutes a violation of the no unmeasured confounding assumption. A natural question is whether the common practice of simply adjusting for the mismeasured confounder is justifiable. In this article, we answer this question in the affirmative and demonstrate that in many realistic scenarios not covered by previous literature, adjusting for the mismeasured confounders reduces bias compared to not adjusting. 
\end{abstract}
\section{Introduction}
In observational studies, researchers often only have access to mismeasured versions of the true confounders, rendering unbiased estimation of causal effects impossible without additional assumptions. Nevertheless, the observed mismeasured version of the true confounder can be utilized in various ways. One typical assumption on the mismeasured confounder is that it is independent of all other variables conditional on the true confounder. The conventional wisdom has been that adjusting for such an observed but mismeasured confounder leads to less bias than no adjustment whatsover \citep{GREENLAND1980}, and there are now many analytical results that rigorously verify the conventional wisdom under certain assumptions. This problem is well-studied in the statistical literature, and the observed variable is often referred to as a nondifferentially mismeasured confounder. \cite{Ogburn2013} study the case of ordinal unmeasured confounders, \cite{Gabriel2022} studies dichotomization of continuous unmeasured confounders, \cite{Ogburn2012}, \cite{Pena2020}, and \cite{Sjolander2022} study binary unmeasured confounders, and \cite{Pena2021} studies general discrete unmeasured confounders. Each of the above papers provides sufficient conditions or examples/counterexamples of when the estimate of a causal effect adjusting for the observed mismeasured confounder comes closer to the true causal effect than when using the crude, unadjusted estimate. In this paper, we present an interpretable sufficient condition for ordinal and continuous confounders to reduce bias, with no restriction on the support of the unmeasured and mismeasured confounders.  In a different vein, \cite{Ding2017} demonstrate that adjusting for an instrumental variable can amplify bias, under some general conditions. We draw on this work to develop our results tailored to our setting.

Also related to our work is the vast literature on auxiliary variables with special causal structure. Such variables should be correlated with unmeasured confounders, and examples include negative control outcomes, negative control exposures, and second treatments. 
These variables possess qualitative similarities with the mismeasured confounders that we study and have been used for many purposes. \cite{rosenbaum1989} and \cite{Lipsitch2010} use negative controls to detect bias, \cite{Rosenbaum2006} and \cite{Chen2023a} use second treatments to learn about differential effects, and
the proximal causal inference literature utilizes negative control outcomes, exposures and/or mismeasured versions of confounders
to obtain consistent effect estimates or bounds \citep{TchetgenTchetgen2014, Kuroki2014, Miao2018IdentifyingConfounder, Park2024SingleControl, Ghassami2023}. The proxy literature, however, requires assumptions on the variability of the true confounder compared to the proxy. In the discrete case, if the unmeasured confounder has more levels than the proxy, the proxy methodology cannot be applied. We tackle the less ambitious goal of decreasing bias by imposing interpretable monotonicity assumptions but do not restrict the variability of the unmeasured confounder relative to its measured counterpart.
\begin{figure}
    \centering
    \begin{tikzpicture}[node distance=1cm, scale = 0.8]
    \tikzstyle{var} = [draw, circle, minimum width=0.8cm, minimum height=0.8cm]
    
    
        \node[var] (A) at (0,0) {$A$};  
        \node[var] (Y) at (3,0) {$Y$};  
        \node[var] (U) at (1.5,2) {$U$}; 
        \node[var] (C) at (1.5,0.7) {$C$};  
    
    \draw[->] (A) -- (Y);
    \draw[->] (U) -- (C);
    \draw[->] (U) -- (Y);
    \draw[->] (U) -- (A);
    
\end{tikzpicture}
    \caption{The relationships between treatment $A$, unmeasured confounder $U$ and its mismeasured version $C$, and outcome $Y$.}
    \label{fig:dag}
\end{figure}
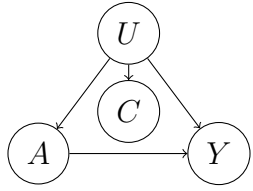
\section{Background and Notation}
In our setting, $U$ is an unmeasured, univariate variable that confounds the relationship between a binary treatment $A$ and outcome $Y$. However, we measure an imperfect proxy for $U$, call it $C$, such that $C$ is only affected by $U$. Thus, $A \indep C \mid U$ and $Y \indep C \mid A, U$. For a directed acyclic graph depiction, see Figure \ref{fig:dag}. We will refer to $C$ as a proxy or mismeasured confounder throughout the manuscript. For ease of exposition, we omit measured confounders $X$, keeping in mind that they can be straightforwardly accommodated with a slight adjustment of the assumptions. All of our assumptions can be interpreted as conditional on $X = x$, and would need to hold for all $x$. We assume the existence of potential outcomes $Y(1)$ and $Y(0)$ that correspond to the outcome that would have been observed if, possibly contrary to fact, a subject had received treatment $A = 1$ and $A = 0$, respectively. The vector $(U, C, A, Y(1), Y(0))$ is assumed to be sampled iid from an unknown distribution $F$. $f(a)$ will denote the density or probability mass function of a random variable $A$ evaluated at $a$, and $f(a \mid b)$ the density conditional on $B = b$. $\mathcal{A}$ will denote the support of a random variable $A$. We will study whether the unadjusted estimate is more or less biased than the adjusted for $C$ estimate for the true effect. 

\section{Qualitative similarity between the mismeasured and unmeasured confounder}
\subsection{Assumptions}
We first introduce the causal assumptions needed so that had $U$ been measured, the causal quantities of interest would be identified.
\begin{assumption}[Causal assumptions]
\label{assumption: causal}
 (i) $Y = Y(1)A + Y(0)(1 - A)$, (ii) For some $\epsilon > 0$, $\epsilon \leq P(A = 1 \mid U = u) \leq 1 - \epsilon$ for all $u$, (iii) $\{Y(1), Y(0)\} \indep A \mid U$.
\end{assumption}
Under these causal assumptions, it is well known that $E(Y(a))$ is identified by $\int E(Y \mid A = a, u) dF(u) = E[AY / P(A = 1 \mid U)]$. We will require some assumptions on the outcome regressions $E(Y \mid A = a, u)$ and propensity score $P(A = 1 \mid u)$, which may need to be verified for a posited $U$ through subject matter expertise. Similar assumptions have also been invoked in \cite{VanderWeele2008}, \cite{Chiba2009},\cite{Ogburn2012},\cite{Ogburn2013}, and \cite{Ding2017}.
\begin{assumption}[Monotonicity] 
\label{assumption: mono}
We require the following conditions:
\begin{enumerate}[label=\roman*)]
\item (outcome regression monotonicity) $E(Y \mid A = a, U = u)$ is non-decreasing in $u$ for $a = 0, 1$.
\item (propensity score monotonicity) $P(A = 1 \mid U = u)$ is non-decreasing in $u$.
\end{enumerate}
\end{assumption}

\begin{remark}
Though we have presented the assumptions in terms of non-decreasing monotonicity, the main attenuation result will continue to hold if one or both montonicity assumptions are flipped to non-increasing. 
\end{remark}

The next assumption formally encapsulates the assumptions implied by the directed acyclic graph in Figure \ref{fig:dag}, and is typically referred to as nondifferential mismeasurement \citep{Ogburn2012}. 
\begin{assumption}[Nondifferential mismeasurement]
\label{assumption:nondifferential}
$C \indep (A,Y) \mid U$.
\end{assumption}
The next two assumptions are the most crucial. They impose specific forms of positive dependence between the unmeasured confounder and the measured $C$. Before laying out the assumptions, we introduce two definitions.

\begin{defn}[\cite{Lehmann1966}]
We say that a random variable $A$ is positively (resp. negatively) regression dependent on $B$ if $P(A \geq a \mid B = b)$ is non-decreasing (resp. non-increasing) in $b$ for all $a \in \mathcal{A}$. If $P(A \geq a \mid B = b, C = c)$ is non-decreasing (resp. non-increasing) in $b$ for all $a \in \mathcal{A}$, we say $A$ is positively (resp. negatively) regression dependent on $B$ given $C = c$.
\end{defn}
\begin{defn}[\cite{Karlin1956, Lehmann1966}]
 Let $a, a' \in \mathcal{A}$ and $b, b' \in \mathcal{B}$. We say that a random variable $A$ has positive (resp. negative) likelihood ratio dependence on $B$ if $f(a' \mid b')f(a \mid b) \geq (\text{resp.} \leq) f(a \mid b')f(a' \mid b)$ when $a' > a, b' > b$. If $f(a' \mid b', c)f(a \mid b, c) \geq (\text{resp.} \leq) f(a \mid b', c)f(a' \mid b, c)$ when $a' > a, b' > b$, we say that $A$ has positive (resp. negative) likelihood ratio dependence on $B$ given $C = c$.
\end{defn}
The first of the two assumptions is based on regression dependence. 
\begin{assumption}[Regression Dependence]
\label{assumption:prds}
(i) $U$ is positive (negative) regression dependent on $C$, (ii) $U$ is positive (negative) regression dependent on $C$ given $A = a$, for $a = 0,1$. 
\end{assumption}
Part (i) is relatively easy to interpret, though it places restrictions on $U \mid C$ rather than $C \mid U$. Part (ii) requires that further conditioning on the treatment status does not break the regression dependence, which is harder to interpret. Under the special case where $U$ is binary, part (ii) of Assumption \ref{assumption:prds} is implied by part (i) when the mismeasurement is nondifferential. This is stated formally below:
\begin{lemma}
\label{lemma: u binary prds}
If $U$ is binary, $C \indep A \mid U$ as in Assumption \ref{assumption:nondifferential}, and Assumption \ref{assumption:prds}(i) holds, then Assumption \ref{assumption:prds}(ii) holds.
\end{lemma}
The second of the two assumptions is based on likelihood ratio dependence. 
\begin{assumption}[Monotone likelihood ratio dependence]
\label{assumption: mlr}
$C$ has positive or negative likelihood ratio dependence on $U$.
\end{assumption}
Assumption \ref{assumption: mlr} turns out to be a stronger assumption than \ref{assumption:prds} when the mismeasurement is nondifferential, which we formalize in the next proposition. Although it is stronger than \ref{assumption:prds}, assumption \ref{assumption: mlr} is attractive in that it places a restriction on the distribution of $C \mid U$, which is more natural than reasoning about the three distributions $U \mid C$, $U \mid A = 1, C$, and $U \mid A = 0, C$. 
\begin{proposition}
\label{prop: mlr implies prd}
If $C \indep A \mid U$ as in Assumption \ref{assumption:nondifferential}, then Assumption \ref{assumption: mlr} implies Assumption \ref{assumption:prds}.
\end{proposition}

\begin{remark}
Although we have presented the assumptions \ref{assumption:prds} and \ref{assumption: mlr} allowing either positive or negative dependence, in our proofs we focus on the positive case without loss of generality. In practice, we would expect $C$ to be positively dependent on $U$.
\end{remark}

We now verify that the relationship between the mismeasured $C$ and the treatment and outcome is qualitatively similar to that of $U$. This result was derived for binary confounders in \cite{Ogburn2012} and \cite{Sjolander2022}. 
\begin{lemma}
\label{lem: proxy_u_similarity}
Under the positive sign versions of assumptions \ref{assumption: mono} and \ref{assumption:prds}, and assumption \ref{assumption:nondifferential},
\begin{align*}
&P(A = 1 \mid C = c), \ E(Y \mid A = 0, C= c), \ E(Y \mid A = 1, C= c)\text{ are non-decreasing in } c.
\end{align*}
\end{lemma}
Equipped with this lemma, we can proceed to the main attenuation result.
\section{Main Result}
We first introduce a helpful lemma that has been invoked in the previous literature.
\begin{lemma}[\cite{Esary1967}, Theorem 3.1]
\label{lem: covariance_inequality}
Let $f(\cdot)$ and $g(\cdot)$ be functions, each with $K$ real-valued arguments, which are both non-decreasing in each of their arguments. If $U=\left(U_1, \ldots, U_K\right)$ is a multivariate random variable with $K$ mutually independent components, then $\operatorname{cov}\{f(U), g(U)\} \geq 0$.
\end{lemma}

The next two lemmas determine the direction of the bias of the adjusted mean compared to the true counterfactual mean as well as the relationship between the unadjusted and adjusted means.
\begin{lemma}
\label{lemma: adjusted larger than true}
Under assumptions \ref{assumption: causal}, the positive sign versions of \ref{assumption: mono},  and \ref{assumption:nondifferential}, 
\begin{equation*}
    \int E(Y \mid A = 1, C = c) f(c) dc \geq E[Y(1)] \text{ and } \int E(Y \mid A = 0, C = c) f(c) dc \leq E[Y(0)]. 
\end{equation*}
\end{lemma}
\begin{lemma}
\label{lemma: unadjusted larger than adjusted}
Under the positive sign versions of assumptions \ref{assumption: mono}, \ref{assumption:prds}, and \ref{assumption:nondifferential},
\begin{equation*}
    E(Y \mid A = 1) \geq \int E(Y \mid A = 1, C = c) f(c) dc   \text{ and } E(Y \mid A = 0) \leq \int E(Y \mid A = 0, C = c) f(c) dc . 
\end{equation*}
\end{lemma}
The bias attenuation property then immediately follows from the two previous lemmas.
 \begin{theorem}
 \label{thm: attenuation}
Let $\mu_a^{\text{unadj}} = E(Y \mid A = a)$ and $\mu_a^{\text{adj}} = \int  E(Y \mid A = a, C =c) f(c) dc$ for $a = 0,1$, and let $h, g$ be any non-decreasing functions. Under assumptions \ref{assumption: causal}, the positive sign versions of \ref{assumption: mono}, \ref{assumption:nondifferential}, and the positive sign version of \ref{assumption:prds} or \ref{assumption: mlr},
\begin{equation}
\label{eq: atteunation}
h[g\{\mu_1^{\text{unadj}}\} - g\{\mu_0^{\text{unadj}}\}] \geq h[g\{\mu_1^{\text{adj}} \} - g\{ \mu_0^{\text{adj}} \}] \geq h[g\{E[Y(1)] \}-g\{E[Y(0)]\}].
\end{equation}
 \end{theorem}
 \begin{proof}
Lemmas \ref{lemma: adjusted larger than true} and \ref{lemma: unadjusted larger than adjusted} imply that $\mu_1^{\text{unadj}} \geq \mu_1^{\text{adj}} \geq E[Y(1)]$ and $-\mu_0^{\text{unadj}} \geq -\mu_0^{\text{adj}} \geq -E[Y(0)]$ when \ref{assumption: causal}-\ref{assumption:prds} hold. Applying the function $g$ to all quantities, adding them together, and then applying $h$ implies the result. We can swap Assumptions \ref{assumption: mlr} and \ref{assumption:prds} since the former implies the latter under assumption \ref{assumption:nondifferential} by Proposition \ref{prop: mlr implies prd}.
 \end{proof}
\begin{remark}
We make several comments. First, when the signs of the monotonicity of the outcome regression and propensity score in Assumption \ref{assumption: mono} are both non-increasing rather than non-decreasing in $u$, Equation \ref{eq: atteunation} continues to hold. If one is non-increasing while the other is non-decreasing, the inequalities in Equation \ref{eq: atteunation} must be flipped, but attenuation continues to hold. Second, for different choices of $g$ and $h$, Theorem $\ref{thm: attenuation}$ implies attenuation on the difference, ratio, and odds ratio effect scales. 
Third, in the supplementary material, we demonstrate that attenuation holds for the effect on the treated. Also, similar to \cite{Ogburn2013} and \cite{Ding2017}, if the unmeasured confounder $U$ is multivariate with independent components, and the monotonicity conditions hold for each component of $U$, attenuation will hold. Finally,
we point out that Assumptions \ref{assumption: mono}-\ref{assumption:prds} together have testable implications. Namely, that $P(A = 1 \mid C=c)$  and $E(Y \mid A = a , C=c)$ are non-decreasing in $c$, and that $\mu_1^{\text{adj}} \leq \mu_1^{\text{unadj}}$ and $\mu_0^{\text{unadj}} \leq \mu_0^{\text{adj}}$. 
\end{remark}

\section{Examining Assumptions \ref{assumption:prds} and \ref{assumption: mlr}}
The utility of Theorem \ref{thm: attenuation} rests entirely on the plausibility and interpretability of 
assumptions \ref{assumption:prds} and \ref{assumption: mlr}. Assumptions \ref{assumption:prds} and \ref{assumption: mlr} are classical notions of positive dependence of random variables, and the former is often invoked in the multiple testing literature as a sufficient condition for false discovery rate control \citep{Benjamini2001}. However, of the assumptions required to establish Theorem \ref{thm: attenuation}, \ref{assumption:prds} and \ref{assumption: mlr} are the hardest to interpret. First, for Assumption \ref{assumption:prds}, a structural causal model in the spirit of \cite{Pearl2009} (and as displayed in Figure \ref{fig:dag}) would posit $U$ having a direct influence or arrow into $C$, not vice versa. The condition that $P(C \geq c \mid U = u)$ is non-decreasing in $u$ would be quite natural, but is not equivalent to Assumption \ref{assumption:prds} in general. Through the next three propositions, we provide one specific but common situation under which Assumption \ref{assumption:prds} holds, and two general settings under which \ref{assumption: mlr} holds, which, as the reader may recall, implies Assumption \ref{assumption:prds} under Assumption \ref{assumption:nondifferential}. Then, we provide several examples that fall under these general settings.

\begin{proposition}[$C$ non-decreasing in $U$ model]
\label{prop: non-decreasing in u}
Suppose $C = \chi(U)$ where $\chi(u)$ is some deterministic function that is non-decreasing in $u$. Then Assumption \ref{assumption:prds} holds. Moreover, if $C^* = \chi^*(C)$ for some $C$ that satisfies Assumption \ref{assumption:prds} and $\chi^*$ non-decreasing, then Assumption \ref{assumption:prds} holds for $C^*$ as well. 
\end{proposition}

\begin{proposition}[Log-concave additive noise, Example 12 in \cite{Lehmann1966}]
\label{prop: log-concave additive noise}
Suppose $C = U + \epsilon$, where $\epsilon$ is random noise that is independent of $U$. If $f_\epsilon$, the density of $\epsilon$, is log-concave, then Assumption \ref{assumption: mlr} holds.
\end{proposition}

\begin{proposition}[Exponential families, \cite{Borges1963}]
\label{prop: nep}
Suppose the conditional density of $C \mid U$ follows $f(c \mid u) = h(c) \exp(\eta(u) \times T(c) - A(u))$, where $\eta$ is a non-decreasing function and $T$ is a strictly increasing function. Then Assumption \ref{assumption: mlr} holds.
\end{proposition}

All of the scenarios in the previous three propositions place realistic and interpretable restrictions on the distribution $C \mid U$. Importantly, in constrast to previous work, they place no restrictions on the distribution or support of $U$. The main distinction between the conditions is that the latter condition imposes a distributional restriction, while the first two impose almost-sure restrictions. Thus, the first two are more amenable to the case where $C$ is truly a mismeasured version of $U$, while the latter is also amenable to the case where $C$ is a negative control outcome or treatment known to be correlated with $U$. Although we have outlined very general conditions under which attenuation will be satisfied, it is instructive to look a few specific examples, some of which have been discussed in previous works on bias attenuation. Detailed justification for each example is contained in the supplementary material.

\begin{example}[Normal-Normal]
Suppose $(U, C)$ are bivariate normal and standard normals marginally, with correlation $\rho > 0$. This situation satisfies Assumption \ref{assumption: mlr}.
\end{example}
\begin{example}[Binary-Binary]
Suppose $U, C$ are binary. Then Assumptions \ref{assumption:prds} and \ref{assumption: mlr} hold if $U$ and $C$ are nonnegatively correlated. If they are negatively correlated, then non-increasing versions of Assumptions \ref{assumption:prds} and \ref{assumption: mlr} hold. The reasoning was discussed in the remark following Proposition \ref{prop: mlr implies prd}. This case was originally investigated in \cite{Ogburn2012}.

\end{example}

\begin{example}[Binary Regression]
Suppose $U \sim F$ and $C \mid U$ either follows a probit or logistic regression. This is a special case of the exponential family, so Assumption \ref{assumption: mlr} is satisfied.

\end{example}
\begin{example}[Additive Noise Differential Privacy]
To satisfy certain privacy constraints, independent noise is often added to the true measurements, such as in the 2020 U.S. Census. Two of the most common mechanisms are additive Gaussian and Laplacian noise \citep{Dwork2013}. Both the Gaussian and Laplacian densities are log-concave, so Assumption \ref{assumption: mlr} is satisfied.
\end{example}

\begin{example}[Coarsened/Binned Variable]
Suppose $U \sim F$ comes from some unknown distribution $F$ and $C$ has $K$ levels. Suppose $\mu_k < \mu_{k+1}$ are some fixed numbers for all $k = 1,\ldots,K-1$ and $\mu_0$ and $\mu_K$ are the essential infimum and supremums of $U$. If $C = k$ when $U \in [\mu_{k-1}, \mu_k]$, then $C$ is a non-decreasing function of $U$, so Assumption \ref{assumption:prds} holds. This scenario is common in practice. Continuous confounders like age, education, or income are often binned into ranges. When $K = 2$, $C = \mathbbm{1}\{U > \theta\}$ for some threshold $\theta$, which was investigated in \cite{Gabriel2022} under specific models for the unmeasured confounder $U$. 
\end{example}

The latter two examples correspond to situations where $C$ is a mismeasured version of $U$, whereas for the previous three, $C$ could plausibly be a proxy of $U$, such as a negative control outcome or exposure. In the next section, we compare our results to those from the previous literature.

\section{A closer comparison to previous literature}
\subsection{Binary $U$, Binary $C$}
We have given an alternative proof that demonstrates that for any binary unmeasured confounder $U$ and nondifferentially mismeasured $C$, adjusting for $C$ attenuates the bias when Assumption \ref{assumption: mono} holds. Such results were originally obtained in \cite{Ogburn2012} using a different proof strategy. Additional attenuation results for this scenario without assuming Assumptions \ref{assumption: mono} were derived in \cite{Pena2020}.
\subsection{A dichotomized unmeasured confounder}
We have generalized results from \cite{Gabriel2022}, who considered a dichotomized confounder under specific distributions of $U \mid A$ and outcome models $E(Y \mid A, U)$. We have shown that under Assumption \ref{assumption: mono}, when $C$ is a binned version of $U$, attenuation holds. Besides providing some specific examples of when attenuation by dichotomization is guaranteed, \cite{Gabriel2022} also wrote down several counterexamples of when dichotomization fails to attenuate bias. In the supplementary material, we examine those counterexamples and demonstrate that in all cases, either Assumption \ref{assumption: mono} (i) or (ii) fails to hold.

\subsection{Ordinal $U$ and $C$}
Although Assumptions \ref{assumption:prds} and \ref{assumption: mlr} apply to ordinal unmeasured confounders, the connection between these assumptions and the tapered distribution assumption introduced in \cite{Ogburn2013} is not obvious. In this subsection, we investigate the relationships between the assumptions. Suppose $U$ and $C$ are both ordinal with $K$ levels, labeled $1, \ldots, K$ without loss of generality. We state the tapered assumption from \cite{Ogburn2013} below.

\begin{assumption}[Tapered misclassification - \cite{Ogburn2013}]
\label{assumption: tapered}
Let 
$p_{ij} = P(C = i \mid U = j)$. 
The misclassification probabilities are tapered, which means that when 
$j \leq k \leq i$, 
$p_{ij} \leq  p_{ik}$ and 
$p_{ji} \leq p_{ki}$, and when 
$i \leq j \leq k$, and 
$p_{ij} \geq  p_{ik}$ and $p_{ji} \geq p_{ki}$.
\end{assumption}
Borrowing the language of \cite{Ogburn2013}, Assumption \ref{assumption: tapered} ``means that
the probability of correct classification is at least as great as the probability of misclassification into any
one level and that, for a fixed level $i$ of either $C$ or $U$, the misclassification probabilities are non-increasing in each direction away from $i$." This means that viewing the $p_{ij}$ as entries of a matrix, starting at a diagonal entry $p_{ii}$, the probabilities decrease as you move away from $p_{ii}$ horizontally and vertically in either direction. One notable difference between the tapered property and positive regression or monotone likelihood ratio dependence is that the latter two are automatically satisfied if $C$ is independent of $U$ whereas the former is not. We now state a result relating Assumption \ref{assumption: tapered} to the flipped version of Assumption \ref{assumption:prds}(i). 

\begin{proposition}
\label{prop: tapered implies prds}
Let $U$ and $C$ be ordinal with $K$ levels, and suppose Assumption \ref{assumption: tapered} holds. Then $P(C \geq i \mid U = j)$ is non-decreasing in $j$. In other words, tapered misclassification implies $C$ is positive regression dependent on $U$.
\end{proposition}
Interestingly, the proof we present is valid if $P_{C \mid U}$ is tapered only in the horizontal direction \citep{Wang2014}, which means when $j \leq  k \leq i$, $p_{ij} \leq p_{ik}$,
and when $i \leq  j \leq k$, $p_{ij} \geq p_{ik}$. Although Proposition \ref{prop: tapered implies prds} uncovers a nice connection between tapered misclassification and positive regression dependence, the positive regression dependence from Assumption \ref{assumption:prds}(i) is in the opposite direction. Unlike monotone likelihood ratio dependence, positive regression dependence and tapered misclassification are not symmetric properties. Figure \ref{fig: ordinal venn diagram} displays the relationships between the three notions of dependence for ordinal variables in a Venn diagram. We give examples of all four possibilities in the supplementary material. For ordinal variables, there exist distributions $C \mid U$ that satisfy Assumptions \ref{assumption:prds} and/or \ref{assumption: mlr} but not Assumption \ref{assumption: tapered}, and vice versa. Therefore, though our assumptions do not generalize the tapered assumption, we have expanded the set of ordinal distributions $C \mid U$ for which attenuation rigorously holds.
\begin{figure}
    \centering
\begin{tikzpicture}[scale = 1]
    \fill[gray!0] (0,0) rectangle (4,3);
    \fill[gray!0] (1.5,1.5) circle (1);
    \fill[gray!0] (2.5,1.5) circle (1);
    \node at (1,1.5) {\scriptsize Taper};
    \node at (3,1.5) {\scriptsize MLR$^*$};
    \node at (0.5,2.5) {\scriptsize PRD};
    \draw (0,0) rectangle (4,3);
    \draw (1.5,1.5) circle (1);
    \draw (2.5,1.5) circle (1);
\end{tikzpicture}
    \caption{A Venn diagram outlining the relationships between tapered (Taper), positively regression dependent (PRD), and positive monotone likelihood ratio dependent (MLR) distributions. The star indicates that MLR is a symmetric property.}
    \label{fig: ordinal venn diagram}
\end{figure}
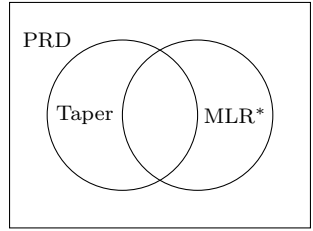
One case that our results do not directly apply to is when $U$ and $C$ are discrete with more than two levels, but with no ordering, as studied in \cite{Pena2021}. This is because our dependence assumptions require $U$ and $C$ to have natural orderings. 

\section{Discussion}
Adjusting for a nondifferentially mismeasured confounder is conventional wisdom that is practiced widely. Previous analytic results demonstrate that the conventional wisdom is justified in several realistic settings. We have substantially expanded the set of scenarios where the conventional wisdom is justified, including for additive noise and coarsening. Promising future directions would be to obtain similarly general results when the measurement error is not non-differential. Another question is whether general attenuation results can be obtained for all continuous linear functionals of $E(Y \mid A, X)$ by imposing a monotonicity assumption on the Riesz representer of the functional.

\section*{Acknowledgments}
The authors thank Professor Betsy Ogburn for helpful suggestions.

\bibliographystyle{apalike}
\bibliography{references}

\newpage
\appendix
\renewcommand{\thetheorem}{\thesection.\arabic{theorem}}
\renewcommand{\thelemma}{\thesection.\arabic{lemma}}
\renewcommand{\theexample}{\thesection.\arabic{example}}
\setcounter{theorem}{0}
\setcounter{lemma}{0}
\setcounter{example}{0}

\section{Proofs of results from Section 3}
\subsection{Proof of Lemma \ref{lemma: u binary prds}}
\begin{proof}[Proof of Lemma \ref{lemma: u binary prds}]
 Since $U$ is binary, we only need to check whether $P(U = 1 \mid A = 1, C = c)$ is increasing in $c$. Observe that by definition and Assumption \ref{assumption:nondifferential},
\begin{equation*}
\begin{aligned}
&P(U = 1 \mid A = 1, C = c) \\&= \frac{P(A = 1 \mid U = 1, C = c) P(U = 1\mid C = c)}{P(A = 1 \mid U=1, C = c)P(U = 1 \mid C = c) + P(A = 1 \mid U=0, C = c)P(U = 0, C = c)} \\&= \frac{P(A = 1 \mid U = 1) P(U = 1\mid C = c)}{P(A = 1 \mid U=1)P(U = 1 \mid C = c) + P(A = 1 \mid U=0)P(U = 0 \mid C = c)} \\ &= \frac{a l(c)}{a l(c) + b(1 - l(c))},
\end{aligned}
\end{equation*}
where we set $a = P(A = 1 \mid U = 1)$, $b = P(A = 1 \mid U=0)$, and $l(c) = P(U = 1 \mid C = c)$. All three quantities lie between 0 and 1. Fix some $c' > c$. $l(c)$ increases in $c$, so going from $\frac{a l(c)}{a l(c) + b(1 - l(c))}$ to $\frac{a l(c')}{a l(c') + b(1 - l(c'))}$ increasing the numerator by $a[l(c')-l(c)] \geq 0$, but the denominator by changes by $a[l(c')-l(c)]+b[l(c)-l(c')] \leq a[l(c')-l(c)]$. Of course, the entire quantity $\frac{a l(c)}{a l(c) + b(1 - l(c))}\leq 1$, so adding a number to the numerator and adding a smaller number to the denominator must increase its value. The argument for the $A = 0$ case is identical; just replace $a$  with $P(A = 0 \mid U = 1)$ and $b$ with $P(A = 0 \mid U=0)$.
\end{proof}
\subsection{Proof of Proposition \ref{prop: mlr implies prd}}
Before proceeding to the proof of Proposition \ref{prop: mlr implies prd}, we state two simple lemmas that we believe to be well-known.
\begin{lemma}
\label{lemma: mlr bidirection}
If Assumption \ref{assumption: mlr} holds for $C \mid U$, it also holds for $U \mid C$.
\end{lemma}

\begin{lemma}
\label{lemma: mlr implies prds}
If Assumption \ref{assumption: mlr} holds for $C \mid U$, then $P(C \geq c \mid U = u)$ is increasing in $u$.
\end{lemma}
Now we present the proof of the proposition. 
\begin{proof}[Proof of Proposition \ref{prop: mlr implies prd}]
First, by Lemma \ref{lemma: mlr bidirection}, we know that $U \mid C$ satisfies the monotone likelihood ratio property under Assumption \ref{assumption: mlr}. By Lemma \ref{lemma: mlr implies prds}, this implies directly that $P(U \geq u \mid C = c)$ is increasing in $c$, i.e. Assumption \ref{assumption:prds}(i). We next demonstrate that the monotone likelihood ratio condition for $U \mid C$ continues to hold after conditioning on $A = a$ under Assumption \ref{assumption:nondifferential}. Fix some $u' \geq u, c' \geq c$ (all are in the supports of $U$ or $C$). Observe that 
\begin{equation*}
\begin{aligned}
&f(c' \mid u')f(c \mid u)  \geq f(c' \mid u)f(c \mid u') \implies  \\ & f(c' \mid u')f(c \mid u)f(a \mid u')f(a \mid u)  \geq f(c' \mid u)f(c \mid u') f(a \mid u')f(a \mid u) \implies \\ & f(c',a \mid u')f(c,a \mid u)  \geq f(c',a \mid u)f(c,a \mid u') \implies \\ & f(u' \mid a, c') f(a, c') / f(u') \times f(u \mid a, c) f(a, c) / f(u) \geq f(u \mid a, c') f(a, c') / f(u) \times f(u' \mid a, c) f(a, c) / f(u') \implies \\ & f(u' \mid a, c') f(a, c')  \times f(u \mid a, c) f(a, c) \geq f(u \mid a, c') f(a, c')  \times f(u' \mid a, c) f(a, c) \implies \\ & f(u' \mid a, c') \times f(u \mid a, c)  \geq f(u \mid a, c')  \times f(u' \mid a, c).
\end{aligned}
\end{equation*}
The second and fourth lines are by Bayes rule, and the third is by Assumption \ref{assumption:nondifferential}. The cancellations in the last two lines are valid because $u,u'$ and $c,c'$ are in the supports of $U$ and $C$ respectively, and by Assumption \ref{assumption: causal}(ii) combined with Assumption \ref{assumption:nondifferential}, which together guarantee $0 < P(A = a \mid C = c) < 1$ for all $c$. We have thus shown that under Assumption \ref{assumption:nondifferential}, monotone likelihood ratio dependence of $C \mid U$ (equivalently $U \mid C$) implies monotone likelihood ratio dependence of $U$ on $C$ given $A = a$. Another application of Lemma \ref{lemma: mlr implies prds} implies that $P(U \geq u \mid A = a, C = c)$ is non-decreasing in $c$ for $a = 0,1$ which is exactly Assumption \ref{assumption:prds}(ii).

\end{proof}

\subsection{Proofs of Lemmas \ref{lemma: mlr bidirection} and \ref{lemma: mlr implies prds}}
\begin{proof} [Proof of Lemma \ref{lemma: mlr bidirection}]
 This is a classical result but we provide a proof for completeness. Let $u' \ge u$ and $c'\ge c$ (all in the support of $U$ or $C$). The monotone likelihood ratio condition for $C\cond U$ tells us that $f(c' \mid u')f(c \mid u) \geq f(c \mid u')f(c' \mid u)$. By Bayes' rule, we can rewrite this inequality as 
 \begin{equation*}
\begin{aligned}
f(c', u')f(c, u) / [f(u') f(u)] \geq f(c, u')f(c', u) / [f(u') f(u)]
\end{aligned}
 \end{equation*}
 We can multiply both sides by $[f(u') f(u)] / [f(c') f(c)]$ to obtain 
  \begin{equation*}
\begin{aligned}
& f(c', u')f(c, u) / [f(c') f(c)] \geq f(c, u')f(c', u) / [f(c') f(c)] \implies \\ &f(u' \mid c')f(u \mid c)  \geq f(u' \mid c)f(u \mid c'),
\end{aligned}
 \end{equation*}
 as desired.
\end{proof}

\begin{proof}
    [Proof of Lemma \ref{lemma: mlr implies prds}]
    This lemma is well-known; see Lemma 3.4.2 on page 78 of \cite{Lehmann2022} for a proof.
    
\end{proof}

\subsection{Proof of Lemma \ref{lem: proxy_u_similarity}}
We introduce an auxiliary lemma that is required for the proof of Lemma \ref{lem: proxy_u_similarity}.

\begin{lemma}
\label{lemma: essential supremum mono}
Under Assumption \ref{assumption:prds}, the essential supremums of $U \mid C = c$ and $U \mid (A = a, C = c)$ are non-decreasing in $c$.
\end{lemma}
\begin{proof}[Proof of Lemma \ref{lemma: essential supremum mono}]
The arguments for $U \mid C = c$ and $U \mid (A = a, C = c)$ are identical up to invoking part (i) or (ii) of Assumption \ref{assumption:prds}, so we only state the argument for $U \mid C = c$. The proof is by contradiction. Suppose there is a $c' > c$ such that the essential supremum for $U \mid C = c'$, call it $\overline{u}(c')$ is less than the essential supremum for $U \mid C = c$, call it $\overline{u}(c)$. The precise definitions are $\overline{u}(c') := \inf\{m \in \mathbb{R} : P(U > m \mid C = c') = 0\}$ and $\overline{u}(c) := \inf\{m \in \mathbb{R} : P(U > m \mid C = c) = 0\}$. If $\overline{u}(c') < \overline{u}(c')$, we can pick some number $m'$ such that $\overline{u}(c') < m' < \overline{u}(c')$. Moreover, by definition, $P(U \geq m' \mid C = c') = 0$ and $P(U \geq m' \mid C = c) > 0$. This contradicts Assumption \ref{assumption:prds}(i).
\end{proof}

\begin{proof}[Proof of Lemma \ref{lem: proxy_u_similarity}]
We present an argument similar to that of Theorem 2 from \cite{Ding2017}. Let $\overline{u}(c)$ and $\underline{u}(c)$ be the essential supremum and infimum of $U \mid C = c$, respectively. For both arguments, we use integral notation, which can be interpreted as a summation if $U$ is discrete. We also use derivative notation with respect to $U$, which can be interpreted as the difference between function values at two consecutive points in the support of $U$ if $U$ is discrete. For the first result, by the law of total probability and integration/summation by parts, 
\begin{equation*}
\begin{aligned}
P(A = 1 \mid C = c) &= \int P(A = 1 \mid U = u, C = c) P(U = u \mid C = c) du \\ &= \int P(A = 1 \mid  U = u) P(U = u \mid C = c) du  \\ &= \int P(A = 1 \mid  U = u) \frac{d}{du}P(U \leq u \mid C = c) du \\ &=  P(A = 1 \mid  U = u)P(U \leq u \mid C = c) |^{\overline{u}(c)}_{\underline{u}(c)}  \\ &- \int   \frac{d}{du} P(A = 1 \mid  U = u)P(U \leq u \mid C = c) du \\ &=  P(A = 1 \mid  U =  \overline{u}(c))  \\ &- \int   \frac{d}{du} P(A = 1 \mid  U = u)P(U \leq u \mid C = c) du 
\end{aligned}
\end{equation*}
 By Assumption \ref{assumption: mono}(ii), we know $\frac{d}{du} P(A = 1 \mid  U = u) \geq 0 $, and by assumption \ref{assumption:prds}, $P(U \leq u \mid C = c)$ is decreasing in $c$. We also know $\overline{u}(c)$ is non-decreasing in $c$ by Lemma \ref{lemma: essential supremum mono} which implies $P(A = 1 \mid  U =  \overline{u}(c))$ is as well. Thus, $P(A = 1 \mid  U =  \overline{u}) - \int   \frac{d}{du} P(A = 1 \mid  U = u)P(U \leq u \mid C = c) du $ is non-decreasing in $c$, and thus $P(A = 1 \mid C = c)$ is as well. 
 
 Let $\overline{u}(a,c)$ and $\underline{u}(a,c)$ be the essential supremum and infimum of $U \mid A = a, C = c$, respectively. For the second result, observe that 
\begin{equation*}
\begin{aligned}
E(Y \mid A = a, C= c) &= \int E(Y \mid A = a, U = u, C = c) P(U = u \mid A = a, C = c) du \\ &= \int E(Y \mid A = a, U = u) P(U = u \mid A = a, C = c) du \\ &= \int E(Y \mid A = a, U = u) \frac{d}{d u}P(U \leq u \mid A = a, C = c) du \\ &= E(Y \mid A = a, U = u)P(U \leq u \mid A = a, C = c) |^{\overline{u}(a,c)}_{\underline{u}(a, c)} \\ & - \int \frac{d}{d u} E(Y \mid A = a, U = u) P(U \leq u \mid A = a, C = c) du \\ &= E(Y \mid A = a,  U = \overline{u}(a,c)) \\ & - \int \frac{d}{d u} E(Y \mid A = a, U = u) P(U \leq u \mid A = a, C = c) du
\end{aligned}
\end{equation*}
By assumption \ref{assumption: mono}(i), we know $\frac{d}{du} E(Y \mid A = a,  U = u) \geq 0 $, and by assumption \ref{assumption:prds}, $P(U \leq u \mid A = a, C = c)$ is decreasing in $c$.  We also know $\overline{u}(a,c)$ is non-decreasing in $c$ by Lemma \ref{lemma: essential supremum mono} which implies $E(Y \mid A = a, U = \overline{u}(a,c))$ is as well. Thus, $E(Y \mid A = a, U = \overline{u}(a,c)) - \int \frac{d}{d u} E(Y \mid A = a, U = u) P(U \leq u \mid A = a, C = c) du$ is non-decreasing in $c$, and thus $E(Y \mid A = a, C = c)$ is as well. 
\end{proof}

\section{Proofs of results from Section 4}
\subsection{Proof of Lemma \ref{lemma: adjusted larger than true}}
\begin{proof}
The proof follows nearly identical steps as in Theorem 1 of \cite{VanderWeele2008}. The adjusted mean for $A = 1$ based on the proxy $C$ is 
\begin{equation*}
\begin{aligned}
     \int & E[Y \mid A = 1, C = c] f(c) dc = \int \int E[Y \mid A = 1, U = u, C = c] f(u \mid A = 1, c) f(c) du dc  \\ &= \int \int E[Y \mid A = 1, U = u, C = c] P(A = 1 \mid u, c)/f(A = 1 \mid c) f(u \mid c) f(c) du dc \\ &= \int E\{ E[Y \mid A = 1, U, C=c] P(A = 1 \mid U, C=c)/P(A = 1 \mid C) \mid C = c \} f(c)  dc \\ & \geq \int E\{ E[Y \mid A = 1, U] \mid C = c \} f(c)  dc \\ &= \int \int E[Y \mid A = 1, U =u, C = c \} f(u\mid c)  f(c) du  dc \\ &= \int \int E[Y \mid A = 1, U =u, C = c \} f(u, c) du  dc \\ &= \int \int E[Y \mid A = 1, U =u\} f(u, c) du  dc \\ &= \int  E[Y \mid A = 1, U =u\} f(u) du = E[Y(1)]  
\end{aligned}
\end{equation*}
The inequality is due to Lemma \ref{lem: covariance_inequality} combined with $E[Y \mid A = a, U = u, C=c] = E[Y \mid A = a, U = u]$ is non-decreasing in $u$ and $P[ A = 1 \mid U = u, C=c] = P[A = a \mid U = u]$ is non-decreasing in $u$ by Assumption \ref{assumption: mono} and \ref{assumption:nondifferential}. Thus, we have established $\int E[Y \mid A = 1, C = c] f(c) dc \geq E[Y(1)] $. Similarly, one can show that $\int E[Y \mid A = 0, C = c] f(c) dc \leq E[Y(0)] $, with the only difference being $P(A = 0 \mid U = u, C = c)$ is non-increasing in $u$. 
\end{proof}

\subsection{Proof of Lemma \ref{lemma: unadjusted larger than adjusted}}
\begin{proof}
The proof is nearly identical to the proof of Lemma \ref{lemma: adjusted larger than true}. 
\begin{equation*}
\begin{aligned}
     & E[Y \mid A = 1]  = \int E[Y \mid A = 1,  C = c] f(c \mid A = 1) dc  \\ &= \int E[Y \mid A = 1,C = c] P(A = 1 \mid c)/P(A = 1)  f(c) dc \\ &=  E\{ E[Y \mid A = 1, C] P(A = 1 \mid C)/P(A = 1)\} \\ & \geq E\{ E[Y \mid A = 1, C]\}  \\ &= \int  E[Y \mid A = 1, C =c\} f(c) dc.   
\end{aligned}
\end{equation*}
The inequality is due to Lemmas \ref{lem: covariance_inequality} and \ref{lem: proxy_u_similarity}. The analogous inequality for $A = 0$ is similarly derived.
\end{proof}

\section{Proofs of results from Section 5}
\subsection{Proof of Proposition \ref{prop: non-decreasing in u}}
\begin{proof}
 Let $C = \chi(U)$, where $\chi$ is a non-decreasing function. Also, let $\text{pre}(c) \equiv \{u \mid  \chi(u) = c\}$ be the pre-image of $c$. Then we can write
\begin{equation}
\label{eq: prd preimage}
\begin{aligned}
P(U \geq u \mid C = c) &= P(U \geq u \mid  \chi(U) = c) \\ &= P(U \geq u \mid U \in  \text{pre}(c)).
\end{aligned}
\end{equation}
Consider any $c' > c$. By a property of functions, preimages of disjoint sets are themselves disjoint. In addition, we assumed that $\chi(u)$ is non-decreasing in $u$. We then can deduce that every element of $\text{pre}(c')$ is greater than every element of $\text{pre}(c)$. This immediately implies that $P(U \geq u \mid U \in  \text{pre}(c))$ is non-decreasing in $c$. This proves part (i) of Assumption \ref{assumption:prds}. 

For part (ii), we can adapt the argument to work when conditioning on $A = 0, 1$. Note that the strict positivity from Assumption \ref{assumption: causal}(ii) along with Assumption \ref{assumption:nondifferential} implies that $P(A = a \mid C)$ is strictly bounded between zero and one, so the conditioning event $(A = a, C)$ is always well-defined, so we can perform the same decomposition of Equation \ref{eq: prd preimage} to get that $P(U \geq u \mid A = a, C = c) = P(U \geq u \mid A = a, U \in  \text{pre}(c))$. Then, due to the fact that the preimages remain disjoint $P(U \geq u \mid A = a, U \in  \text{pre}(c))$ is also non-decreasing in $c$ for $a = 0,1$. 

Next, we consider some $C^* = \chi^*(C)$ for some $C$ that satisfies Assumption \ref{assumption:prds}. We check part (i) for $C^*$ first. Define $\text{pre}^*(c^*) \equiv \{c \mid  \chi^*(c) = c^*\}$.
\begin{align*}
P(U \geq u \mid C^* = c^*) &= P(U \geq u \mid  \chi^*(C) = c^*) \\ &= P(U \geq u \mid C \in  \text{pre}^*(c^*)).
\end{align*}
Now fix any $\Tilde{c}^* > c^*$. By a property of functions, $\text{pre}^*(c^*)$ and $\text{pre}^*(\Tilde{c}^*)$ are disjoint, and all elements in the latter are greater than all elements in the former since $\chi^*$ is non-decreasing. Since $C$ satisfies Assumption \ref{assumption:prds}, $P(U \geq u \mid C = c)$ is non-decreasing in $c$. Thus, for any $\Tilde{b} \in \text{pre}^*(\Tilde{c}^*)$ and any $b \in \text{pre}^*(c^*)$, we have $P(U \geq u \mid C = \Tilde{b}) \geq P(U \geq u \mid C = b)$. Therefore, $P(U \geq u \mid C \in  \text{pre}^*(\Tilde{c}^*)) \geq P(U \geq u \mid C \in  \text{pre}^*(c^*))$ which is equivalent to the statement, $P(U \geq u \mid C^* = \Tilde{c}^*)) \geq P(U \geq u \mid C^* = c^*)$. The arguments for the $A = 0, 1$ cases are identical and thus omitted.
\end{proof}

\subsection{Proof of Proposition \ref{prop: log-concave additive noise}}
\begin{proof}
This is a well-known result. See Example 8.2.1 on page 369 of \cite{Lehmann2022} for a proof. Also, see \cite{Efron1965} for a related result.  
\end{proof}

\subsection{Proof of Proposition \ref{prop: nep}}
\begin{proof}
The statement of the proposition is well-known; see for instance Corollary 3.4.1 on page 74 of \cite{Lehmann2022}. We provide a proof for completeness. We can check the monotone likelihood ratio condition directly. Fix some $u' \geq u, c' \geq c$. Then for $h(c) \neq 0$,
\begin{equation*}
\begin{aligned}
\frac{f(c \mid u')}{f(c \mid u)} &= \frac{h(c)\exp(\eta(u') \times T(c) - A(u'))}{h(c)\exp(\eta(u) \times T(c) - A(u))} \\ &= \frac{\exp(\eta(u') \times T(c) - A(u'))}{\exp(\eta(u) \times T(c) - A(u))} \\ &\propto \frac{\exp(\eta(u') \times T(c))}{\exp(\eta(u) \times T(c))} \\ &= \exp(T(c) \times (\eta(u') - \eta(u))).
\end{aligned}
\end{equation*}
Since $\eta$ is increasing in $u$, $\eta(u') \geq \eta(u)$. $T$ is also increasing in $c$, so increasing $c$ must also increase the final quantity, thereby proving the claim. When $h(c) = 0$, both $f(c \mid u')$ and $f(c \mid u)$ are zero, so the inequality $f(c' \mid u') f(c \mid u) \geq f(c' \mid u) f(c \mid u')$ trivially holds.
\end{proof}

\section{Detailed justification for the examples from Section 5}
\subsection{Normal-Normal}

We will check the monotone likelihood ratio condition directly. Note that $C \mid U = u \sim N(\rho u, 1-\rho^2)$. Then for $u' \geq u$, $c' \geq c$, \begin{align*}
    f(c'\mid u') f(c\mid u) = \frac{1}{2\pi(1-\rho^2)} \exp(-0.5(c' - \rho u')^2/(1-\rho^2)) \exp(-0.5(c - \rho u)^2/(1-\rho^2)) \\ f(c\mid u') f(c'\mid u) = \frac{1}{2\pi(1-\rho^2)} \exp(-0.5(c - \rho u')^2/(1-\rho^2)) \exp(-0.5(c' - \rho u)^2/(1-\rho^2))
\end{align*}
One can verify that the top quantity is larger than the bottom quantity by comparing the exponents, i.e. $-0.5(c' - \rho u')^2/(1-\rho^2) + -0.5(c - \rho u)^2/(1-\rho^2)$ to $-0.5(c' - \rho u)^2/(1-\rho^2) + -0.5(c' - \rho u)^2/(1-\rho^2)$. Subtracting the latter from the former yields $2\rho(u'-u)(c'-c)/(1-\rho^2) \geq 0$.

If $C$ had actually been generated from the model $C = \rho U + \sqrt{1- \rho^2}\epsilon$, where $U \sim N(0,1)$ and $\epsilon \sim N(0,1)$ are independent, then this would match the conditional distribution of a bivariate standard gaussian with correlation $\rho$. It would also satisfy the log-concave additive noise model, since we could take $\sqrt{1- \rho^2}\epsilon$ to be the independent noise, $g(u) = \rho u$ and $h$ the identity function. Gaussian densities are log-concave.

\subsection{Binary-Binary}
We first check Assumption \ref{assumption:prds} directly. Let $p_{ab}=P(U=a,C=b)$, for $a,b \in \{0,1\}$. Note that $P(U \leq 0 \mid C = 1) = P(U = 0 \mid C = 1) = \frac{p_{01}}{p_{11}+p_{01}}$ and $P(U \leq 0 \mid C = 0) = P(U = 0 \mid C = 0) = \frac{p_{00}}{p_{10}+p_{00}}$. Assumption \ref{assumption:prds}(i) would hold if $\frac{p_{00}}{p_{10}+p_{00}} \geq \frac{p_{01}}{p_{11}+p_{01}}$. The covariance between $U$ and $C$ is given by 
\begin{equation*}
\begin{aligned}
p_{11}-(p_{01}+p_{11})(p_{10}+p_{11})&= p_{11}-(p_{01}p_{10}+p_{01}p_{11}+p_{11}p_{10}+p_{11}^2)\\ &= p_{11}(1-p_{01}-p_{10}-p_{11})-p_{01}p_{10} \\ &=p_{11}p_{00}-p_{01}p_{10}.
\end{aligned}
\end{equation*}
If the covariance is nonnegative, then 
\begin{equation*}
    p_{11}p_{00}-p_{01}p_{10} \geq 0 \implies p_{11}p_{00} + p_{00}p_{01} \geq p_{01}p_{10} + p_{00}p_{01} \implies \frac{p_{00}}{p_{10}+p_{00}} \geq \frac{p_{01}}{p_{11}+p_{01}}.
\end{equation*}
So Assumption \ref{assumption:prds}(i) holds if $U$ and $C$ are not negatively correlated.
Lemma \ref{lemma: u binary prds} implies (ii) and (iii) also hold if Assumption \ref{assumption:nondifferential} holds.

Checking Assumption \ref{assumption: mlr}
amounts to checking $P(U = 1 \mid C = 1) / P(U = 1 \mid C = 0) \geq P(U = 0 \mid C = 1) / P(U = 0 \mid C = 0)$. This is equivalent to $p_{11} / p_{10} \geq p_{01} / p_{00}$, which is equivalent to $ p_{11}p_{00}-p_{01}p_{10} \geq 0$, i.e. that $U$ and $C$ are positively correlated.

\subsection{Binary Regression}
 Probit and logistic regression are special cases of the exponential family. This is because we can write
 \begin{align*}
    f(c \mid u) = \exp(c \log(p(u) / (1-p(u)) + \log(1-p(u))),
 \end{align*}
 where $p(u) = \Phi(\tau_0 + \tau_1 u)$ for probit and $p(u) = \text{expit}(\tau_0 + \tau_1 u)$ for logistic. These both fit the exponential family case by taking $A(u) = \log(1-p(u))$, $T(c) = c$ which is strictly increasing in $c$, $h(c) = 1$, and $\eta(u) = \log(p(u) / (1-p(u))$ which is non-decreasing in $u$ since it is a composition of non-decreasing functions.

Also, if $C$ were generated from the latent variable interpretation of probit and $U$ is from a log-concave distribution, we can apply Propositions \ref{prop: non-decreasing in u} and \ref{prop: log-concave additive noise} jointly. Namely, if we take $C^* = \tau_0 + \tau_1 U + \epsilon$, where $\epsilon \sim N(0,1)$ and $C = \mathbbm{1}(C^* > 0)$. So take $\chi(\cdot)$ to be the indicator that the argument is greater than $0$. Then $C$ would a non-decreasing function of $C^*$ which is the sum of $U$ and noise, both of which are from log-concave distributions. Moreover, the latent variable interpretation of probit regression can be extended to the ordinal case, where $C$ is a binned version of $C^*$.

\subsection{Additive Noise Differential Privacy}
The justification is immediate from the fact that the Laplace and Gaussian densities are log-concave.

\subsection{Dichotomized/Coarsened Variable}
This is a special case of the $C$ non-decreasing in $U$ model, by taking $\chi$ to be the coarsening/dichotomizing function that maps $C$ to $\{1,\ldots,K\}$ based on some fixed thresholds $\mu_k$.

\section{Proofs of results from Section 6}
\subsection{Proof of Proposition \ref{prop: tapered implies prds}}
\begin{proof}
    For $i = 1$ and $i = K$, that $P(C\ge i \mid U=j)$ is non-decreasing in $j$ is immediate; for $i = 1$ the quantity is always 1 and for $i = K$, $P(C \geq K \mid U = j) = P(C = K \mid U = j)$ and so the tapered condition implies this quantity is non-decreasing in $j$. Thus, we restrict our attention to $2\le i \le K-1$.
    
    Fixing $i$, we note that Assumption \ref{assumption: tapered} implies that
    $$
         p_{i'j} \le p_{i'(j+1)},
    $$
    for each $i' \ge i$ and $j\le i-1$. Taking the sum over $i'$, we see that $P(C\ge i\mid U=j)$ is increasing from $j=1$ to $i$. \newline
    \indent
    To complete the proof, we must show the same holds when $j$ increases from $i$ to $K$. We similarly invoke Assumption \ref{assumption: tapered} to state that
    $$
        p_{i'j} \ge p_{i'(j+1)},
    $$
    when $K-1 \ge j \ge i$ and $i' < i$. Again, for each such $j$, we can sum over all such $i'$:
    $$
        \sum_{i' < i}p_{i'j} \ge \sum_{i' < i}p_{i'(j+1)} \implies P(C < i \mid U = j) \ge P(C < i \mid U = j+1).
    $$
    Taking the complimentary statement of the right-hand side above, we see that $P(C \ge i \mid U = j) \le  P(C \ge i \mid U = j+1)$ for $j \ge i$ as well. We conclude that $P(C \ge i \mid U = j)$ must be nondecreasing in $j$. For the proof, notice we have only used the ``horizontal" part of the tapered condition, which states that for $j \leq  k \leq i$, $p_{ij} \leq p_{ik}$ and when $i \leq  j \leq k$, and $p_{ij} \geq p_{ik}$.
\end{proof}

\section{Counterexamples from \cite{Gabriel2022}}
\begin{enumerate}
\item We first consider settings 4 and 5 from \cite{Gabriel2022}, with notation adapted to our setting. They set the outcome model to be $E(Y \mid A = a, U = u) = \alpha + \beta a + \gamma u + \delta a u$. In some counterexamples for setting 4, they set $\gamma = 1$ and $\delta = -2.2$. This results in $E(Y \mid A = 1, U = u) = \alpha + \beta  - 1.2 u$ and $E(Y \mid A = 0, U = u) = \alpha + u$, so the outcome models are monotone in $u$ but in different directions, violating Assumption \ref{assumption: mono}(i). Similarly, in setting 5,  they set $\gamma = 1$ and $\delta = -1.8$. One can check that this also leads to  outcome models that are monotone in $u$ but in different directions.

\item We next consider setting 6 of \cite{Gabriel2022}. In this setting, They posit 
\begin{equation*}
\begin{aligned}
& U \mid A = 0 \sim 
\text{Gamma}(1.1, 2) \, U \mid A = 1 \sim 
\text{Gamma}(3, 0.8), \\ & E(Y \mid A = a, U = u) = \alpha + \beta a + \gamma u.
\end{aligned}
\end{equation*}
The gamma parameters are the shape and scale parameters. It is easy to see that Assumption \ref{assumption: mono}(i) is not violated here. This suggests that \ref{assumption: mono}(ii) must be violated. It is possible to obtain a closed form for $P(A = 1 \mid U = u)$ using Bayes' rule. We omit the analytical formula and instead plot it over a range of $u$ values in the left panel of Figure \ref{fig:gabriel counter}. We see that, indeed, the propensity score is not monotone in $u$. 
\item Finally we examine setting 7 from \cite{Gabriel2022}. In this setting, the authors consider an additional covariate variable, which we will call $X$. We did not explicitly write our results with $X$, but claimed a slight tweak to the assumptions would suffice. Explicitly, we would require $E(Y \mid A = a, U = u, X = x)$ to be non-decreasing in $u$ for $a = 0, 1$ and all $x$ as well as $P(A = 1 \mid U = u, X = x)$ to be non-decreasing in $u$ for all $x$. Their setting was as follows: 
\begin{equation*}
\begin{aligned}
& X \sim \text{Bernoulli}(d), \  A \mid X = x\sim \text{Bernoulli}(p_x), \\ & U \mid A = a, X = x \sim 
N(\mu_{ax}, 1),
 \\ & E(Y \mid A = a, U = u, X = x) = \alpha + \beta a + \gamma u + \rho x.
\end{aligned}
\end{equation*}
It is easy to see that the outcome monotonicity will hold for $a = 0,1$ and all $x$ as there are no interactions between $u$ and $a$ or $x$ in the linear outcome model. We now closely examine the propensity model. As in setting 6, we can get a closed form for $P(A = 1 \mid U = u, X = x)$ using Bayes' rule.  We omit the analytical formulas and instead plot $P(A = 1 \mid U = u, X = x)$ over a range of $u$ values for $x = 0,1$ in the right panel of Figure \ref{fig:gabriel counter}. We see that the propensity scores are monotone in $u$ for both values of $x$, but in different directions.

\begin{figure}[htbp]
    \centering
    \includegraphics[width=0.35\textwidth]{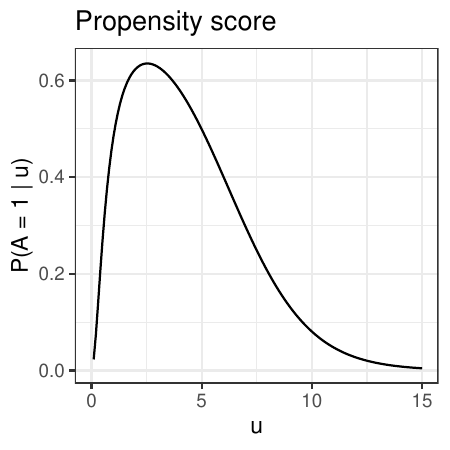}
    \hfill
    \includegraphics[width=0.47\textwidth]{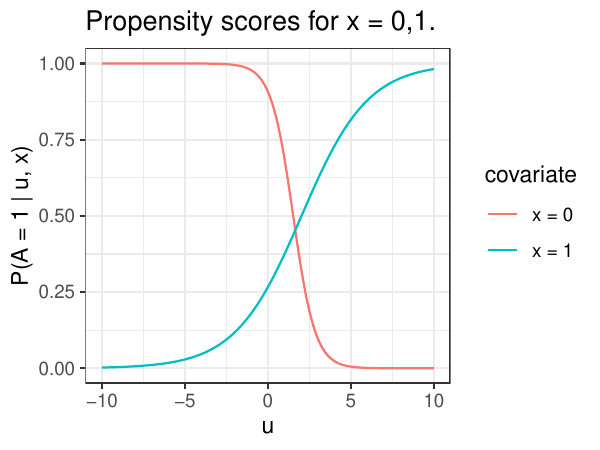}
    \caption{Left: A plot of the propensity score in setting 6 of \cite{Gabriel2022}. Right: A plot of the propensity score in setting 7 of \cite{Gabriel2022} colored by covariate.}
    \label{fig:gabriel counter}
\end{figure}

\end{enumerate}

\section{Examples and counterexamples from Section 6.3}
For these set of examples, we consider the case where $U$ and $C$ have 3 levels, and WLOG, we label these $1,2,3$. The $i,j$ entry of a matrix $P_{A \mid B}$ is simply $P(A = i \mid B = j)$, for $i,j \in \{1,2,3\}$.
\begin{example}[$\text{PRD } C \mid U \nRightarrow \text{PRD } U \mid C$]
Consider the following matrix:
\begin{equation*}
   P_{C \mid U} = \begin{bmatrix}
0.4 & 0.3 & 0.30\\
0.5 & 0.5 & 0.25 \\
0.1 & 0.2 & 0.45 \\
\end{bmatrix}.
\end{equation*}
Suppose the marginal distribution of $U$ is given by the probability vector $[0.2,0.5,0.3]^T$, where the $j$th entry equals $P(U = j)$. One can check that 
\begin{equation*}
   P_{U \mid C} = \begin{bmatrix}
0.25000 & 0.2352941 & 0.07843137 \\
0.46875 & 0.5882353 & 0.39215686 \\
0.28125 & 0.1764706 & 0.52941176
\end{bmatrix}.
\end{equation*}
This distribution does not follow PRD, since $P(U \geq 3 \mid U = 2)  = 0.1764706 < P(U \geq 3 \mid U = 1) = 0.28125$.

\end{example}
\begin{example}[$\text{Tapered } C \mid U \nRightarrow \text{Tapered } U \mid C$]
Consider the following matrix:
\begin{equation*}
   P_{C \mid U} = \begin{bmatrix}
0.39 & 0.32 &  0.3\\
0.31 & 0.37 & 0.31 \\
0.3 & 0.31 & 0.39 \\

\end{bmatrix}.
\end{equation*}
Suppose the marginal distribution of $U$ is given by the probability vector $[0.25,0.4,0.35]^T$, where the $j$th entry equals $P(U = j)$. One can check that 
\begin{equation*}
   P_{U \mid C} = \begin{bmatrix}
0.2950076 & 0.2320359 & 0.2235469 \\
0.3872920 & 0.4431138 & 0.3695976 \\
0.3177005 & 0.3248503 & 0.4068554
\end{bmatrix}.
\end{equation*}
This matrix is not tapered since $P(U = 2 \mid C = 1) = 0.3872920 > P(U = 1 \mid C = 1) = 0.2950076$.
\end{example}

\begin{example}[$C \mid U$ PRD, Tapered, and MLR]
Consider the following matrix:
\begin{equation*}
   P_{C \mid U} = \begin{bmatrix}
0.7 & 0.25 & 0.05 \\
0.25 & 0.5 & 0.25  \\
0.05 & 0.25 & 0.7
\end{bmatrix}.
\end{equation*}

This matrix is tapered by inspection. The MLR can be checked by inspecting three sequences: 1) $\frac{P(C = 1 \mid U = 2)}{P(C = 1 \mid U = 1)} = \frac{0.25}{0.7} \leq \frac{P(C = 2 \mid U = 2)}{P(C = 2 \mid U = 1)} = \frac{0.5}{0.25} \leq \frac{P(C = 3 \mid U = 2)}{P(C = 3 \mid U = 1)} = \frac{0.25}{0.05}$, 2) $\frac{P(C = 1 \mid U = 3)}{P(C = 1 \mid U = 1)} = \frac{0.05}{0.7} \leq \frac{P(C = 2 \mid U = 3)}{P(C = 2 \mid U = 1)} = \frac{0.25}{0.25} \leq \frac{P(C = 3 \mid U = 3)}{P(C = 3 \mid U = 1)} = \frac{0.7}{0.05}$, 3) $\frac{P(C = 1 \mid U = 3)}{P(C = 1 \mid U = 2)} = \frac{0.05}{0.25} \leq \frac{P(C = 2 \mid U = 3)}{P(C = 2 \mid U = 2)} = \frac{0.25}{0.5} \leq \frac{P(C = 3 \mid U = 3)}{P(C = 3 \mid U = 2)} = \frac{0.7}{0.25}$.
\end{example}
\begin{example}[$C \mid U$ PRD, not Tapered, and not MLR]
\begin{equation*}
   P_{C \mid U} = \begin{bmatrix}
1/4 & 1/6 & 1/7\\
1/4 & 1/3 & 1/7  \\
1/2 & 1/2 & 5/7
\end{bmatrix}.
\end{equation*}
This is PRD since $P(C \geq 2 \mid U = 1) = 1/4 + 1/2 \leq P(C \geq 2 \mid U = 2) = 1/3 + 1/2 \leq P(C \geq 2 \mid U = 3) = 1/7 + 5/7$ and $P(C \geq 3 \mid U = 1) = 1/2 \leq P(C \geq 3 \mid U = 2) = 1/2 \leq P(C \geq 3 \mid U = 3) = 6/7$. This matrix is not MLR since $P(C = 2 \mid U = 3) / P(C = 2 \mid U = 2)  = \frac{1/7}{1/3}< P(C = 1 \mid U = 3) / P(C = 1 \mid U = 2) = \frac{1/7}{1/6}$. This matrix is not tapered since $P(C = 3 \mid U = 1) = 1/2 > P(C = 1 \mid U = 1) = 1/4$.
\end{example}

\begin{example}[$C \mid U$ PRD and Tapered, not MLR]
Consider the following matrix:
\begin{equation*}
   P_{C \mid U} = \begin{bmatrix}
0.5 & 0.25 &  0.25 \\
0.25 & 0.5 & 0.25  \\
0.25 & 0.25 & 0.5
\end{bmatrix}.
\end{equation*}
It is easy to see that this matrix is tapered. This matrix is not MLR since $P(C = 2 \mid U = 3) / P(C = 2 \mid U = 2)  = \frac{0.25}{0.5}< P(C = 1 \mid U = 3) / P(C = 1 \mid U = 2) = \frac{0.25}{0.25}$.
    
\end{example}

\begin{example}[$C \mid U$ PRD and MLR, not Tapered]
Consider the following matrix:
\begin{equation*}
   P_{C \mid U} = \begin{bmatrix}
1/6 & 1/6 & 1/6 \\
1/3 & 1/3 & 1/3  \\
1/2 & 1/2 & 1/2
\end{bmatrix}.
\end{equation*}

This matrix is not tapered since $P(C = 3 \mid U = 1) = 1/2 > P(C = 1 \mid U = 1) = 1/6$. $C$ is independent of $U$, so it trivially follows that the matrix is MLR and PRD.

A more interesting example is the following:
\begin{equation*}
   P_{C \mid U} = \begin{bmatrix}
1/3 & 0 & 0 \\
1/6 & 1/4 & 1/4  \\
1/2 & 3/4 & 3/4
\end{bmatrix}.
\end{equation*}
This matrix is not tapered since $P(C = 3 \mid U = 1) = 1/2 > P(C = 1 \mid U = 1) = 1/3$. The MLR can be checked by inspecting three sequences: 1) $\frac{P(C = 1 \mid U = 2)}{P(C = 1 \mid U = 1)} = \frac{0}{1/3} \leq \frac{P(C = 2 \mid U = 2)}{P(C = 2 \mid U = 1)} = \frac{1/4}{1/6} \leq \frac{P(C = 3 \mid U = 2)}{P(C = 3 \mid U = 1)} = \frac{3/4}{1/2}$, 2) $\frac{P(C = 1 \mid U = 3)}{P(C = 1 \mid U = 1)} = \frac{0}{1/3} \leq \frac{P(C = 2 \mid U = 3)}{P(C = 2 \mid U = 1)} = \frac{1/4}{1/6} \leq \frac{P(C = 3 \mid U = 3)}{P(C = 3 \mid U = 1)} = \frac{3/4}{1/2}$, 3) $\frac{P(C = 1 \mid U = 3)}{P(C = 1 \mid U = 2)} = \frac{0}{0} := 1 \leq \frac{P(C = 2 \mid U = 3)}{P(C = 2 \mid U = 2)} = \frac{1/4}{1/4} \leq \frac{P(C = 3 \mid U = 3)}{P(C = 3 \mid U = 2)} = \frac{3/4}{3/4}$.
\end{example}

\section{Effect of Treatment on the Treated}
As discussed in the main manuscript, the attenuation continues to hold for the effect of the treatment on the treated, i.e. $E[Y(1) - Y(0) \mid A = 1]$. The unadjusted, adjusted, and true version of $E[Y(1) \mid A = 1]$, coincide, so we need only compare the corresponding versions of $E[Y(0) \mid A = 1]$. The unadjusted version is $E[Y \mid A = 0]$ since we remain in the setting with no measured confounders. The adjusted version is $\int E[Y \mid A = 0, C = c] f(c \mid A = 1)$. The true $E[Y(0) \mid A = 1]$ can be expressed as $\int E[Y \mid A = 0, U = u] f(u \mid A = 1)$. We now present two analogous lemmas to \ref{lemma: adjusted larger than true} and \ref{lemma: unadjusted larger than adjusted}. The arguments we present below would need to be slightly altered with measured confounders, and the steps are slightly more involved than incorporating measured confounders for the average treatment effect case; see Theorem 2 of \cite{Chiba2009}.
\begin{lemma}
Under assumptions \ref{assumption: causal}, the positive sign versions of \ref{assumption: mono},  and \ref{assumption:nondifferential}, 
\begin{equation*}
\int E[Y \mid A = 0, C = c] f(c \mid A = 1) dc \leq \int E[Y \mid A = 0, U = u] f(u \mid A = 1) du
\end{equation*}
\end{lemma}
\begin{proof}
This proof is nearly identical to that of Theorem 2 from \cite{Chiba2009}. 
\begin{align*}
\int &E[Y \mid A = 0, c] f(c \mid A = 1) dc = \int \int  E[Y \mid A = 0, c, u] f(u \mid A = 0, c) du f(c \mid A = 1) dc \\ &= \int \int  E[Y \mid A = 0, c, u] \frac{P(A = 0 \mid u, c)f(u \mid c)}{P(A = 0 \mid c)} du f(c \mid A = 1) dc \\ &= \int E_{U \mid c}[E[Y \mid A = 0, c, U] P(A = 0 \mid U, c)] f(c \mid A = 1) / P(A = 0 \mid c) dc \\ &\leq \int E_{U \mid c}[E[Y \mid A = 0, c, U]]E_{U \mid c}[ P(A = 0 \mid U, c)] f(c \mid A = 1) / P(A = 0 \mid c) dc \\ &= \int E_{U \mid c}[E[Y \mid A = 0, c, U]] f(c \mid A = 1)  dc \\ &= \int E_{U \mid c}[E[Y \mid A = 0, c, U]] E_{U \mid c}[ P(A = 1 \mid U, c)] f(c \mid A = 1) / P(A = 1 \mid c) dc \\ &\leq \int E_{U \mid c}[E[Y \mid A = 0, c, U] P(A = 1 \mid U, c)] f(c \mid A = 1) / P(A = 1 \mid c) dc \\ &= \int \int E[Y \mid A = 0, c, u] P(A = 1 \mid u, c)] f(u \mid c) f(c \mid A = 1) / P(A = 1 \mid c) dc \\ &= \int \int E[Y \mid A = 0, c, u] P(A = 1 \mid u, c) f(u \mid c) f(c) / P(A = 1) du dc \\ &= \int \int E[Y \mid A = 0, u] P(A = 1 \mid u) f(u,c) / P(A = 1) dc du \\ &= \int \int E[Y \mid A = 0, u] P(A = 1 \mid u) f(c \mid u) f(u) / P(A = 1) dc du \\ &= \int  E[Y \mid A = 0, u] P(A = 1 \mid u) f(u) / P(A = 1) du \\ &= \int  E[Y \mid A = 0, u] f(u \mid A = 1) du. 
\end{align*}
\end{proof}
The first inequality is due to $E[Y \mid A = 0, C = c, U = u] = E[Y \mid A = 0, U = u]$ being non-decreasing in $u$ and $P(A = 0 \mid C = c, U = u) = P(A = 0 \mid U = u)$ being non-increasing in $c$ and Lemma \ref{lem: covariance_inequality}. The second inequality is also due to Lemma \ref{lem: covariance_inequality} as well as $E[Y \mid A = 0, C = c, U = u] = E[Y \mid A = 0, U = u]$ and $P(A = 0 \mid C = c, U = u) = P(A = 0 \mid U = u)$ both being non-decreasing in $u$. The fourth from the last equality is by Assumption \ref{assumption:nondifferential}.

\begin{lemma}
Under the positive sign versions of assumptions \ref{assumption: mono}, \ref{assumption:prds}, and \ref{assumption:nondifferential},
\begin{equation*}
E[Y \mid A = 0] \leq \int E[Y \mid A = 0, C = c] f(c \mid A = 1) dc
\end{equation*}
\end{lemma}
\begin{proof}
\begin{align*}
E[Y \mid A = 0] &= \int E[Y \mid C = c, A = 0] f(C = c \mid A = 0) dc \\ &= \int E[Y \mid C = c, A = 0] \frac{P(A = 0 \mid C = c)}{P(A = 0)} f(c) dc \\ &\leq E_C[E[Y \mid C, A = 0]] \times E_C[\frac{P(A = 0 \mid C = c)}{P(A = 0)}] \\ &= E_C[E[Y \mid C, A = 0]]  \\ &= E_C[E[Y \mid C, A = 0]] E_C[\frac{P(A = 1 \mid C = c)}{P(A = 1)}] \\ &\leq \int E[Y \mid C = c, A = 0] \frac{P(A = 1 \mid C = c)}{P(A = 1)} f(c) dc \\ &= \int E[Y \mid C = c, A = 0] f(c \mid A = 1) dc. 
\end{align*}
The first inequality is due to $E[Y \mid C = c, A = 0]$ being non-decreasing in $c$ and $P(A = 0 \mid C = c)$ being non-increasing in $c$ by Lemma \ref{lem: proxy_u_similarity} and Lemma \ref{lem: covariance_inequality}. The second inequality is also due to Lemma \ref{lem: covariance_inequality} as well as $E[Y \mid C = c, A = 0]$ and $P(A = 1 \mid C = c)$ both being non-decreasing in $c$.
\end{proof}
These lemmas together establish that 
$E[Y \mid A = 0] \leq \int E[Y \mid A = 0, C = c] f(c \mid A = 1) dc \leq \int E[Y \mid A = 0, U = u] f(u \mid A = 1) du$, establishing the attenuation result for the effect on the treated. An analogous proof can be shown for the effect on the control, which is omitted.

\end{document}